\documentclass[11pt,reqno]{amsart} 
\usepackage{amssymb,amsthm,amsmath,epsfig,latexsym}
\usepackage{graphicx}
\usepackage{calc,times}
\usepackage{tikz}
\usetikzlibrary{matrix}
\usepackage[all]{xy}

\usepackage{geometry} 
\begin{document}


\newcommand{\B}{{\mathcal B}}
\newcommand{\dee}{{\mathrm d}}
\newcommand{\bl}{{\mathrm Bl}}
\newcommand{\rees}{{R[It]}}
\newcommand{\fiber}{{\mathcal{F}}(I)}
\newcommand{\cone}{{\mathcal{G}}(I)}
\newcommand{\blowup}{Bl}
\newcommand{\pp}{\mathbb P}
\newcommand{\recip}{{\mathcal{L}}^{-1}_{\A}}
\newcommand{\gas}{{\mathcal{G}}}
\newcommand{\fm}{\mathfrak m}
\newcommand{\xx}{\mathbf x}
\newcommand{\mmbox}[1]{\mbox{${#1}$}}
\newcommand{\proj}[1]{\mmbox{{\mathbb P}^{#1}}}
\newcommand{\Cr}{C^r(\Delta)}
\newcommand{\CR}{C^r(\hat\Delta)}
\newcommand{\affine}[1]{\mmbox{{\mathbb A}^{#1}}}
\newcommand{\Ann}[1]{\mmbox{{\rm Ann}({#1})}}
\newcommand{\caps}[3]{\mmbox{{#1}_{#2} \cap \ldots \cap {#1}_{#3}}}
\newcommand{\N}{{\mathbb N}}
\newcommand{\Z}{{\mathbb Z}}
\newcommand{\R}{{\mathbb R}}
\newcommand{\Tor}{\mathop{\rm Tor}\nolimits}
\newcommand{\Ext}{\mathop{\rm Ext}\nolimits}
\newcommand{\Hom}{\mathop{\rm Hom}\nolimits}
\newcommand{\im}{\mathop{\rm Im}\nolimits}
\newcommand{\rank}{\mathop{\rm rank}\nolimits}
\newcommand{\supp}{\mathop{\rm supp}\nolimits}
\newcommand{\arrow}[1]{\stackrel{#1}{\longrightarrow}}
\newcommand{\CB}{Cayley-Bacharach}
\newcommand{\coker}{\mathop{\rm coker}\nolimits}
\newcommand{\m}{{\frak m}}
\newcommand{\fitt}{{\rm Fitt}}
\newcommand{\C}{{\mathcal{C}}}
\newcommand{\K}{\mathbb K}
\newcommand{\OT}{R}  
\newcommand{\pan}{{\rm Span }}  
\newcommand{\M}{\mathsf M}
\newcommand{\Ima}{{\rm Im}\,}


\makeatletter
\renewcommand*\env@matrix[1][*\c@MaxMatrixCols c]{%
  \hskip -\arraycolsep
  \let\@ifnextchar\new@ifnextchar
  \array{#1}}
\makeatother


\sloppy
\newtheorem{defn0}{Definition}[section]
\newtheorem{prop0}[defn0]{Proposition}
\newtheorem{quest0}[defn0]{Question}
\newtheorem{thm0}[defn0]{Theorem}
\newtheorem{lem0}[defn0]{Lemma}
\newtheorem{corollary0}[defn0]{Corollary}
\newtheorem{example0}[defn0]{Example}
\newtheorem{remark0}[defn0]{Remark}
\newtheorem{prob0}[defn0]{Problem}
\newtheorem{conj0}[defn0]{Conjecture}

\newenvironment{defn}{\begin{defn0}}{\end{defn0}}
\newenvironment{prop}{\begin{prop0}}{\end{prop0}}
\newenvironment{quest}{\begin{quest0}}{\end{quest0}}
\newenvironment{thm}{\begin{thm0}}{\end{thm0}}
\newenvironment{lem}{\begin{lem0}}{\end{lem0}}
\newenvironment{cor}{\begin{corollary0}}{\end{corollary0}}
\newenvironment{exm}{\begin{example0}\rm}{\end{example0}}
\newenvironment{rem}{\begin{remark0}\rm}{\end{remark0}}
\newenvironment{prob}{\begin{prob0}\rm}{\end{prob0}}
\newenvironment{conj}{\begin{conj0}}{\end{conj0}}

\newcommand{\defref}[1]{Definition~\ref{#1}}
\newcommand{\propref}[1]{Proposition~\ref{#1}}
\newcommand{\thmref}[1]{Theorem~\ref{#1}}
\newcommand{\lemref}[1]{Lemma~\ref{#1}}
\newcommand{\corref}[1]{Corollary~\ref{#1}}
\newcommand{\exref}[1]{Example~\ref{#1}}
\newcommand{\secref}[1]{Section~\ref{#1}}
\newcommand{\remref}[1]{Remark~\ref{#1}}
\newcommand{\questref}[1]{Question~\ref{#1}}
\newcommand{\probref}[1]{Problem~\ref{#1}}

\newcommand{\std}{Gr\"{o}bner}
\newcommand{\jq}{J_{Q}}
\def\Ree#1{{\mathcal R}(#1)}

\numberwithin{equation}{subsection}  


\title{Generalized Hamming weights and minimal shifts of Orlik-Terao algebras}
\author{\c{S}tefan O. I. Toh\v{a}neanu}

\dedicatory{To my father, Ion V. Toh\v{a}neanu}

 \keywords{generalized Hamming weights, graded Betti numbers, Orlik-Terao algebra, initial degree. \\ \indent Tohaneanu's Address: Department of Mathematics, University of Idaho, Moscow, Idaho 83844-1103, USA, Email: tohaneanu@uidaho.edu, Phone: 208-885-6234, Fax: 208-885-5843.}

\begin{abstract} In this note we show that the minimum distance of a linear code equals one plus the smallest shift in the first step of the minimal graded free resolution of the Orlik-Terao algebra (i.e., the initial degree of the Orlik-Tearo ideal) constructed from any parity-check matrix of the linear code. We move forward with this connection and we prove that the second generalized Hamming weight equals one or two plus the smallest shift at second step in the minimal graded free resolution of the same algebra. Via a couple of examples we show that this ambivalence is the best result one can get if one uses Orlik-Terao algebras to characterize the second generalized Hamming weight.
\end{abstract}
\maketitle

\section{Introduction}

The minimum distance $d$ of a linear code $\C$ is the most important parameter associated to $\C$ because it controls the error-correction capabilities of $\C$. Unfortunately, in general $d$ is very difficult to calculate; by \cite{Va}, computing $d$ is NP-hard.

Since the early development of what now are known as Algebraic Geometric codes, and all their commutative algebraic relatives (Reed-Solomon codes, Reed-Muller codes, affine variety codes, evaluation codes, etc.), there has been a growing interest into understanding and maybe even identifying the minimum distance with some other algebraic invariants, some more classical than the others. Furthermore, with the beginning of using Gr\"{o}bner bases to error-correct corrupted message received (see \cite{Aug}, \cite{BuPe2}), the attention turned almost entirely towards studying ideals in rings of polynomials in several variables with coefficients in a field.

One of the algebraic invariants that we mentioned above is the {\em initial degree} (some authors call it, the {\em $\alpha$-invariant}) of a positively graded finitely generated module over a standard graded Noetherian ring; this number captures the minimum degree of a nonzero element of the graded module. \cite{GaTo} explored this connection in details, as follows: in \cite{DePe}, and further developed in \cite{To1}, the minimum distance can be calculated recursively by investigating the heights of several ideals constructed from a generator matrix of the linear code. It turns out that one can hide this recursive calculation by computing the initial degree of the Fitting module of $\C$ (see \cite[Theorem 2.1]{GaTo}), or of the associated graded module corresponding to some filtration that occurs for the case of binary linear codes (see \cite[Theorem 3.7]{GaTo}). In the same article, the initial degree of apolar algebras (i.e., inverse systems) of hyperplane arrangements gives a lower bound for the minimum distance (\cite[Proposition 3.9]{GaTo}), and lastly, there is some connection between the minimum distance, the initial degree, and the linear strand of the Orlik-Terao algebra of the line arrangement constructed from a generator matrix of any dimension 3 linear code (\cite[Proposition 3.9]{GaTo}).

It was not an artificial stretch to look at the Orlik-Terao algebra of a hyperplane arrangement, except that \cite{GaTo} looked at the wrong hyperplane arrangement. In this note we will correct this point of view by looking at the Orlik-Terao algebra of the hyperplane arrangement defined by the linear forms dual to the columns of a {\em parity-check} matrix of $\C$ (instead of a generator matrix of $\C$).

In general, the Orlik-Terao algebra is relatively new in the field of commutative algebra (see \cite{OrTe}, \cite{Te}, or \cite{ScTo}; in other references it is called ``the algebra of the reciprocal plane''), and because of its features trying to mimic the celebrated Orlik-Solomon algebra, it seems to be specific only to the theory of hyperplane arrangements. But, by using \cite[Theorem 2.4]{GaSiTo}, this algebra is identified with the {\em fiber cone} (a.k.a. the {\em special fiber}) of an ideal associated to this hyperplane arrangement. This later graded algebra is a familiar presence in algebraic geometry and commutative algebra, since it is in direct connection to the famous ``resolution of singularities'' and ``blowups''. Using this isomorphism, the interested reader can translate all the results in this paper in fiber cone/special fiber terminology.

The point of view mentioned above, and presented in this note, comes as an immediate consequence of the characterization of the generalized Hamming weights using a parity-check matrix. The most important article who looked at this characterization but in connection to a different graded algebra, the Stanley-Reisner ring of the matroid of the parity-check matrix, is \cite{JoVe1}. This algebra is, in spirit, more combinatorial than algebraic, but is much simpler to define and to work with: the Stanley-Reisner ideal is monomial and square-free. Because of this, using \cite[Theorem 5.1]{Hoc}, the main result \cite[Theorem 4.2]{JoVe1} proves that for any $1\leq r\leq k$, where $k$ is the dimension of the linear code, the $r$-th generalized Hamming weight equals the smallest shift occurring at the $r$-th steps (i.e., $r$-th homological degree) in the minimal graded free resolution of the Stanley-Reisner ring of the matroid associated to a parity-check matrix of the code. On the same idea, the very recent work \cite{DiFoKuTo} shows that the generalized Hamming weights can be identified with the smallest degree of a square-free monomial in the corresponding symbolic powers of the same Stanley-Reisner ideal.

Unfortunately, the graded Betti numbers of the Orlik-Terao ideal are not known to benefit of a direct correspondence with the underlying combinatorics. It is relatively immediate to see that the initial degree of the Orlik-Terao ideal equals the minimum distance minus one (see, Theorem \ref{thm3}). But once we move further into the minimal graded free resolution, it becomes quite challenging to understand the syzygies of this ideal. Two main reasons for this are that these ideals are not combinatorial (see \cite[Example 1.4]{ScTo}), nor they are monomial (nor binomial), so they do not benefit from some computational advantages (such as monomial/binomial Gr\"{o}bner bases) that will produce somewhat ``standard'' free resolutions (see, for example, Taylor free resolution of a monomial ideal: \cite[Exercise 17.11]{Ei0}). One other disadvantage is that, to our knowledge, only a couple of articles turned their attention towards studying the free resolutions of Orlik-Terao ideals, in general: \cite[Section 3]{ScTo} presents some results concerning combinatorial first syzygies, and, the more elusive, \cite[Section 4.3]{GaSiTo} where the handling of syzygies of the Rees ideal (hence, of the Orlik-Terao ideal) is hidden behind some powerful homological techniques and results. As we will see later, even if we pass to the initial ideals of Orlik-Terao ideals, which are the Stanley-Reisner ideals of broken-circuit complexes, the graded Betti numbers depend on the monomial ordering of the variables (see Example \ref{exm12}). For $r=2$, in Theorem \ref{thm13} we show that we can place $d_2(\C)$ between $t+1$ and $t+2$, where $t$ is the smallest shift occurring at second step in the minimal graded free resolution of the Orlik-Terao ideal of $\mathcal A(\mathcal C^{\perp})$. Unfortunately, we cannot do better than these inequalities: Examples \ref{exm10} and \ref{exm15} show that both bounds can be attained. In Remark \ref{lastrem} we discuss briefly about higher generalized Hamming weights, and we conclude with a conjecture about a lower bound of these.

We started this article with the goal of identifying the generalized Hamming weights with homological invariants of yet another (commutative) algebra. In the end, we came to the conclusion that despite a very strong connection with the Stanley-Reisner algebras of two similar simplicial complexes constructed from the circuits of the same hyperplane arrangement, the shifts in the minimal graded free resolution of the Orlik-Terao algebra of that hyperplane arrangement have an unclear behaviour that needs to unravelled.

We end the introduction by mentioning \cite{MaMaSu}, where the similar concept of initial degree, but of a non-homogeneous ideal, comes in direct connection with the minimum distance. In \cite[Proposition 2]{MaMaSu}, the error-capability of the linear code $\C$ plus one equals the minimum degree of an element of a Gr\"{o}bner basis of some binomial ideal. This ideal is constructed from a generator matrix of $\C$, and takes into account the multiplication table of the base field.

\section{Minimum distance and initial degree of Orlik-Terao ideals}

\subsection{Linear codes.} Let $\mathbb K$ be any field, and let $n$ be a positive integer. A {\em linear code} $\C$ is a linear subspace of $\mathbb K^n$. $n$ is the {\em length} of $\C$, and $k:=\dim_{\mathbb K}(\C)$ is the {\em dimension} of $\C$.

For any vector ${\bf w}\in\mathbb K^n$, the {\em weight} of ${\bf w}$, denoted $wt({\bf w})$, is the number of nonzero entries in ${\bf w}$. The {\em minimum distance} of $\C$ is the integer: $$d=d(\C):=\min\{wt({\bf c})|{\bf c}\in\C\setminus\{{\bf 0}\}\}.$$ $n$, $k$, and $d$ are called the {\em parameters} of $\C$, and $\C$ is called an {\em $[n,k,d]$-linear code}.

A {\em generator matrix} of $\C$ is any $k\times n$ matrix whose rows form a basis for $\C$. If $G$ is such a matrix, then any element ${\bf c}$ of $\C$ (called {\em codeword}) is a linear combination of the rows of $G$, so ${\bf c}={\bf v}G$, for some ${\bf v}\in\mathbb K^k$. Using simple linear algebra we obtain that the transpose of any such ${\bf c}$ is in the kernel of an $(n-k)\times n$ matrix, called {\em parity-check matrix} of $\C$. Also, a code is called {\em non-degenerate} if any generator matrix $G$ doesn't have a zero column.

If $\C^{\perp}$ denotes the orthogonal dual of $\C$ in $\mathbb K^n$, then $\C^{\perp}$ is an $[n,n-k,d(\C^{\perp})]$-linear code, called the {\em dual code} of $\C$, and it has a generator matrix $H$, where $H$ is a parity-check matrix of $\C$.

\begin{rem}\label{rem1} Let $\C$ be an $[n,k,d]$-linear code with a generator matrix $G$. We have the following properties:
\begin{itemize}
  \item[(i)] If $G$ is in standard form, i.e., $G=\left[I_k|P\right]$, then $H:=\left[-P^T|I_{n-k}\right]$ is a parity-check matrix of $\C$. By the row reduction algorithm, and possibly after a permutation of the columns of $G$ (which produces a generator matrix of an {\em equivalent code} to $\C$; none-the-less, this ``new'' linear code has the same parameters as $\C$), any matrix $G$ can be brought to standard form.
  \item[(ii)] $H$ has a column consisting all of zeros iff $P$ has a row consisting all of zeros iff $d=1$. This is saying that $\C^{\perp}$ is non-degenerate iff $d\geq 2$.
\end{itemize}
\end{rem}

\medskip

\subsubsection{Generalized Hamming weights.} \label{sec:genHam} Let $\mathcal C$ be an $[n,k,d]$-linear code. Let $\mathcal D\subseteq \mathcal C$ be a subcode, which is just a linear subspace of $\C$. The support of $\mathcal D$ is $$Supp(\mathcal D):=\{i |\exists (x_1,\ldots,x_n)\in \mathcal D \mbox{ with }x_i\neq 0\}.$$ Let $m(\mathcal D):=|Supp(\mathcal D)|$ be the cardinality of the support of $\mathcal D$.

If ${\bf c}=(c_1,\ldots,c_n)\in\C$ is a codeword, similarly define the {\em support of ${\bf c}$} to be $supp({\bf c}):=\{i|c_i\neq 0\}$.

For any $r=1,\ldots,k$, the {\em $r$-th generalized Hamming weight of $\mathcal C$} is the positive number $$d_r(\mathcal C):=\min_{\mathcal D\subseteq \mathcal C,\,\dim\mathcal D=r}m(\mathcal D).$$ By convention, $d_0(\mathcal C) = 0$.

By \cite[Theorem 2]{We}, we have the following crucial result. Let $\mathcal C$ be an $[n,k]$-linear code with parity check matrix $H$. If $[n]:=\{1,\ldots,n\}$ and $i\in [n]$, denote by $H_i$ the $i$-th column of $H$. Let $1\leq r\leq k$ be an integer. Then

$$d_r(\mathcal C)=\min_{I\subseteq[n]}\{|I|\,|\, |I|-\dim_{\mathbb K}{\rm Span}_{\mathbb K}\{H_i\,|\, i\in I\}\geq r\}.$$

\begin{rem}\label{rem1.1} We have the following.

\begin{itemize}
\item[(i)] $d_1(\mathcal C)$ equals the minimum distance $d$ of $\mathcal C$; this is because all subcodes of $\C$ of dimension 1 are the linear spans (i.e., scalar multiples) of the nonzero codewords of $\C$.
\item[(ii)] $\C$ has minimum distance $d$ if and only if $d$ is the maximum integer such that any $d-1$ columns of any parity-check matrix of $\C$ are linearly independent.
\item[(iii)] The version of Wei's formula that uses a generator matrix $G$ of an $[n,k]$-linear code $\C$ instead, is the following: $d_r(\mathcal C)$ equals $n$ minus the maximum number of columns of $G$ that span a $k-r$ dimensional vector subspace of $\mathbb K^k$.
\end{itemize}
\end{rem}

\medskip

\subsection{Hyperplane arrangements.} Let $G$ be a $k\times n$ matrix with entries in a field $\mathbb K$; and suppose it doesn't have any column consisting entirely of zeroes. Let $R:=\mathbb K[x_1,\ldots,x_k]$ denote the graded ring of polynomials with the standard grading given by the degree; the zero polynomial has any degree.

To the $i$-th column $[a_{1,i},\ldots,a_{k,i}]^T$ of $G$ we associate the {\em dual linear form} $\ell_i:=a_{1,i}x_1+\cdots+a_{k,i}x_k\in R$. Then, for each $i=1,\ldots,n$ consider the hyperplane $H_i:=V(\ell_i)\subset\mathbb K^k$; this way we obtain a {\em (multi) hyperplane arrangement}. The prefix ``multi'' (that we will omit writing it) is justified by the possibility that maybe some of the columns of $G$ are proportional.

Suppose now that $G$ is a generator matrix of an $[n,k,d]$-linear code $\C$. Denote with $\mathcal A(\C)$ the hyperplane arrangement constructed from $G$ as we described above. This construction is reversible (modulo, ``monomial equivalence of linear codes''): starting with a hyperplane arrangement of $n$ hyperplanes in $\mathbb K^k$ and of rank $k$, by recording the coefficients of each defining linear form, one can form a $k\times n$ matrix of rank $k$ that can be thought of as a generator matrix of a linear code; this matrix is sometimes called {\em the coefficients matrix} of the hyperplane arrangement. If one permutes or rescales the columns of a matrix, on obtains the same hyperplane arrangement if we think of this matrix as the coefficients matrix, and one obtains different, yet equivalent, linear codes if we think of this matrix as a generator matrix.

Any hyperplane arrangement, or its coefficients matrix, comes with a matroid, and the connections between the combinatorics of this matroid and the properties of the associated linear code are very strong; see \cite{JuPe} for extended details.

\medskip

\subsubsection{The Orlik-Terao algebra.} Let $\mathcal A(\C)$ be the hyperplane arrangement associated to any generator matrix of an $[n,k,d]$-linear code $\C$. Suppose we fixed the linear forms $\ell_i\in R$. Let $S:=\mathbb K[y_1,\ldots,y_n]$ and consider the ring epimorphism $$\gamma:S\rightarrow R,\, \gamma(y_i)=\ell_i, 1\leq i\leq n.$$ The kernel of $\gamma$, denoted here $F({\mathcal C})$, is an ideal of $S$, minimally generated by $n-k$ linear forms in $S$. This is often called the \emph{relation space} of $\mathcal A(\C)$. The matrix of the coefficients of the standard basis of $F({\mathcal C})$ is an $(n-k)\times n$ matrix which is the generating matrix of the dual code $\mathcal C^{\perp}$ of $\C$.

Now let us define the ring epimorphism $$\phi: S\rightarrow \mathbb K[1/\ell_1,\ldots, 1/\ell_n],\, \phi(y_i)=1/\ell_i, 1\leq i\leq n.$$ {\em The Orlik-Terao algebra} is ${\rm OT}(\mathcal C):=\K[y_1,\ldots,y_n]/\ker(\phi)$. $\ker(\phi)$ is called {\em the Orlik-Terao ideal}, and it is denoted ${\rm IOT}(\mathcal C)$.

In these pages we will look at the Orlik-Terao algebra of $\mathcal A(\C^{\perp})$. If $d=1$, then by Remark \ref{rem1} (ii), any generator matrix of $\C^{\perp}$ will have a column consisting all of zeros. In this case we cannot construct the Orlik-Terao algebra since there is no hyperplane dual to the zero vector.

\begin{rem}\label{rem2} With the above notations we have the following.

\begin{itemize}
  \item[(i)] It is well known (see \cite{OrTe, Te, ScTo}) that ${\rm IOT}(\mathcal C)$ is generated by
$$\partial(a_1y_{i_1}+\cdots+a_uy_{i_u}):=\sum_{j=1}^{u}a_jy_{i_1}\cdots\widehat{y_{i_j}}\cdots y_{i_u},$$ where $a_1y_{i_1}+\cdots+a_uy_{i_u}, a_j\neq 0$ is an element in the relations space $F({\mathcal C})$.
  \item[(ii)] If the generator matrix is $G=[a_{i,j}], 1\leq i\leq k, 1\leq j\leq n$, then \cite[Proposition 2.6]{DeGaTo} implies that the Orlik-Terao ideal can be obtained as the colon ideal ${\rm IOT}(\C) = J : (y_1\cdots y_n)^{\infty}$, where $J$ is generated by the maximal minors of the following matrix
\[
 \left[\begin{array}{cccc}a_{1,1}y_1& \cdots &a_{1,n}y_n\\ \vdots  & \ddots & \vdots \\ a_{k,1}y_1& \cdots  &a_{k,n}y_n \\
 \hline
 1 & \cdots & 1\end{array}\right].
\]
\item[(iii)] \cite[Theorem 4]{PrSp} shows that the generators obtained at (i) from circuits in the matroid of a generator matrix of $\C$ form a universal Gr\"{o}bner basis for ${\rm IOT}(\mathcal C)$. With this, the initial ideal of ${\rm IOT}(\mathcal C)$ equals the Stanley-Reisner ideal of the ``broken-circuit complex'', which is similar to the Stanley-Reisner ideal considered in \cite{JoVe1}. For example, if $\{1,2,3\}$ is a circuit of the matroid (i.e., the first, second, and third columns of the matrix are linearly dependent), then $y_1y_2y_3$ is an element of the Stanley-Reisner ideal of the matroid simplicial complex. But also $\{1,2\}$, $\{1,3\}$ and $\{2,3\}$ are broken circuits, and depending on the chosen ordering of the variables, $y_1y_2$, or $y_1y_3$, or $y_2y_3$ belong to the Stanley-Reisner ideal of the corresponding broken-circuit complex.
\end{itemize}

\end{rem}
\medskip

The {\em initial degree} of a finitely generated graded module $M=\oplus_{i\geq 0}M_i$ over any positively graded ring, denoted $\alpha(M)$, is the smallest $i$ for which $M_i\neq 0$; in other words, it is the smallest degree of a generator of $M$. Sometimes, this is called the $\alpha$-invariant. With this definition we have the following

\begin{thm}\label{thm3} Let $\C$ be an $[n,k,d]$-linear code such that $d\geq 2$. Then $$d=\alpha({\rm IOT}(\C^{\perp}))+1.$$
\end{thm}
\begin{proof} Let $H$ be a parity-check matrix of $\C$. By Remark \ref{rem1.1} (ii), any $d-1$ columns of $H$ are linearly independent, and there exist $d$ columns of $H$ that are linearly dependent.

Using \cite{JoVe1} or \cite{JuPe} language, the matroid of $H$ (or the hyperplane arrangement $\mathcal A(\C^{\perp})$) has the smallest circuit being of size $d$. This circuit comes from a dependency $a_1\ell_{i_1}+\cdots+a_d\ell_{i_d}=0, a_j\neq 0$ among $d$ defining linear forms for $\mathcal A(\C^{\perp})$, and this dependency corresponds to the element $D:=a_1y_{i_1}+\cdots+a_dy_{i_d}\in F({\mathcal C}^{\perp})$.

Immediately we have $\partial(D)$ is an element of degree $d-1$ of ${\rm IOT}(\C^{\perp})$. But, by Remark \ref{rem2} (i), this ideal is generated by $\partial$ of dependencies corresponding to dependent sets, and since there are no smaller size dependent sets of the matroid of $H$, we have $\alpha({\rm IOT}(\C^{\perp}))=d-1$.
\end{proof}

\begin{exm}\label{exm8} For all the calculations in this example, we are using \cite{GrSt}.

Let $\C$ be a linear code over $\mathbb F_3$ with generator matrix
$$G=\left[\begin{array}{ccccccc}1&0&0&1&1&1&1\\ 0&1&0&0&1&1&0\\ 0&0&1&0&0&2&1\end{array}\right].$$ The corresponding parity-check matrix is $$H=\left[\begin{array}{ccccccc}2&0&0&1&0&0&0\\ 2&2&0&0&1&0&0\\ 2&2&1&0&0&1&0\\ 2&0&2&0&0&0&1\end{array}\right].$$

The defining linear forms of $\mathcal A(\C^{\perp})$ in $R=\mathbb F_3[x_1,x_2,x_3,x_4]$ are
$$\ell_1=2(x_1+x_2+x_3+x_4), \ell_2=2(x_2+x_3), \ell_3=x_3+2x_4, \ell_4=x_1, \ell_5=x_2, \ell_6=x_3, \ell_7=x_4.$$

Let $S:=\mathbb F_3[y_1,\ldots,y_7]$. Then, by using Remark \ref{rem2} (ii), the Orlik-Terao ideal ${\rm IOT}(\C^{\perp})$  has generators
\begin{eqnarray}
&&\partial(y_3-y_6+y_7):= y_3y_6-y_3y_7+y_6y_7,\nonumber\\
&&\partial(y_2+y_5+y_6):= y_2y_5+y_2y_6+y_5y_6,\nonumber\\
&&\partial(y_1-y_2+y_4+y_7):= y_1y_2y_4+y_1y_2y_7-y_1y_4y_7+y_2y_4y_7,\nonumber\\
&&\partial(y_2+y_3+y_5+y_7):= y_2y_3y_5+y_2y_3y_7+y_2y_5y_7+y_3y_5y_7,\nonumber
\end{eqnarray} and some elements of degree four which are boundaries of relations corresponding to any linear dependency of five defining linear forms; $\mathcal A(\C^{\perp})$ has rank $7-3=4$, so any five of the defining linear forms are linearly dependent. Also note that $$\partial(y_2+y_3+y_5+y_7)=(y_3+y_7)\partial(y_2+y_5+y_6)-(y_2+y_5)\partial(y_3-y_6+y_7),$$ so this element will not show up in the set of minimal generators.

It is clear that $\alpha({\rm IOT}(\C^{\perp}))=2$, and therefore the minimum distance of $\C$ is $d=2+1=3$.
\end{exm}

\section{Second generalized Hamming weight and syzygies of Orlik-Terao ideals}

\subsection{Minimal graded free resolution.} Let $I$ be a homogeneous ideal of $S=\mathbb K[y_1,\ldots,y_n]$, the standard graded ring of polynomials in variables $y_1,\ldots,y_n$ with coefficients in a field $\mathbb K$.

Let $$0\rightarrow {\bf F}_p\stackrel{\phi_p}\longrightarrow \cdots \stackrel{\phi_3}\longrightarrow {\bf F}_{2}\stackrel{\phi_2}\longrightarrow {\bf F}_1\stackrel{\phi_1}\longrightarrow {\bf F}_0=S\stackrel{\phi_0}\longrightarrow S/I\rightarrow 0$$ be the minimal graded free resolution of $S/I$, as a graded $S$-module. For each $0\leq i\leq p$, we have ${\bf F}_i=\oplus_j S(-j)^{\beta_{i,j}(S/I)}$ as a graded free $S$-module. The numbers $\beta_{i,j}(S/I)$ are called {\em graded Betti numbers}, and for each $i$, except for finitely many, all $\beta_{i,j}(S/I)$ are zero; if $i=0$, we have $\beta_{0,0}(S/I)=1$ and $\beta_{0,j}(S/I)=0$, for all $j\neq 0$. The length of the minimal free resolution is called the {\em projective dimension}, i.e., $p={\rm pdim}_S(S/I)$.

For $1\leq i\leq p$, we say that ${\bf F}_i$ {\em sits at step $i$ in the free resolution of $S/I$} (or, {\em sits in $i$-th homological degree of $S/I$}), and an element of $\ker(\phi_{i-1})\subseteq{\bf F}_i$ is called an {\em $(i-1)$-th syzygy of $I$}.

For each $1\leq i\leq p$, denote $$t_i(I):=\min\{j|\beta_{i,j}(S/I)\neq 0\}.$$ Note that $t_1(I)=\alpha(I)$. Because of the minimality of the free resolution, we have that the sequence of these minimal shifts, $t_1(I),\ldots,t_p(I)$, is strictly increasing and $t_1(I)>0$. Similarly, one can define the sequence of maximal shifts, $T_1(I),\ldots, T_p(I)$. Clearly $T_i(I)\geq t_i(I)$, for $i=1,\ldots,p$.

The value $\max\{T_i(I)-i|i=1,\ldots,p\}$ is the {\em (Castelnuovo-Mumford) regularity}, denoted ${\rm reg}(S/I)$.

The importance of the minimal and maximal shifts in the free resolution is transparent in the famous Multiplicity Conjecture (due to Herzog, Huneke and Srinivasan): if $S/I$ is Cohen-Macaulay (i.e., ${\rm pdim}_S(S/I)={\rm ht}(I)=p$), then $$\frac{\prod_{i=1}^pt_i(I)}{p!}\leq\deg(I)\leq \frac{\prod_{i=1}^pT_i(I)}{p!}.$$

\medskip

In this section we will focus on $t_2(I)$.

\medskip

\begin{exm}\label{exm10} Let us investigate Example \ref{exm8} a bit further. The circuits of the matroid of the matrix $H$ are

$$\{3,6,7\}, \{2,5,6\}, \{1,2,4,7\}, \{2,3,5,7\}, \{1,2,3,4,5\}, \{1,2,3,4,6\},$$ $$\{1,3,4,5,6\}, \{1,3,4,5,7\}, \{1,4,5,6,7\},$$ and the minimal graded free resolution of the Stanley-Reisner ring $S/I$ is

$$0\rightarrow S^7(-7)\rightarrow S^2(-5)\oplus S^{13}(-6)\rightarrow S^2(-3)\oplus S^2(-4)\oplus S^5(-5)\rightarrow S\rightarrow S/I\rightarrow 0.$$ Here, $S=\mathbb K[y_1,\ldots,y_7]$ and $I$ is generated by the square-free monomials $\displaystyle \prod_{i\in C}y_i$, for any $C$ circuit in the above list.

By \cite[Theorem 4.2]{JoVe1}, the generalized Hamming weights of $\C$ are the smallest shifts at each step in the resolution: $$d_1(\C)=3, d_2(\C)=5, d_3(\C)=7.$$

On the other hand, the minimal graded free resolution of the Orlik-Terao algebra ${\rm OT}(\C^{\perp})$ is

$$0\rightarrow S^3(-6)\rightarrow S(-4)\oplus S^5(-5)\rightarrow S^2(-2)\oplus S(-3)\oplus S(-4)\rightarrow S\rightarrow {\rm OT}(\C^{\perp})\rightarrow 0.$$ Observe that for $r=1,2,3$, $$t_r({\rm IOT}(\C^{\perp}))+1=d_r(\C).$$
\end{exm}

\medskip

A standard technique to study any homogeneous ideal is to turn the attention towards its initial ideal. The initial ideals of Orlik-Terao ideals are Stanley-Reisner ideals of the corresponding broken-circuit complexes (see, \cite{PrSp}), and they were successfully used in some of the proofs in \cite{ScTo}. But, as we will see in the next example, they cannot help to calculating $d_2(\C)$.

\begin{exm}\label{exm12} We will look again at Examples \ref{exm8} and \ref{exm10}. In \cite{GrSt}, the default monomial order is the graded reverse lexicographic order. We will use the same, but we will order the variables of $S$ differently each case.

Above we listed the basis of all circuits, and to obtain the broken-circuits, we have to remove the ``smallest'' element under a fixed ordering of $\{1,\ldots,7\}$ from each circuit.

\begin{itemize}
  \item If we order the variables $y_1>y_2>y_3>y_4>y_5>y_6>y_7$, then this ideal is $$I_1:=\langle y_3y_6, y_2y_5, y_1y_2y_4, y_1y_4y_5y_6, y_1y_3y_4y_6\rangle,$$ with $S/I_1$ having the minimal graded free resolution
      $$0\rightarrow S^3(-6)\rightarrow S^2(-4)\oplus S^5(-5)\rightarrow S^2(-2)\oplus S(-3)\oplus S^2(-4)\rightarrow S\rightarrow S/I_1\rightarrow 0.$$
  \item If we order the variables $y_1>y_2>y_3>y_4>y_6>y_7>y_5$, then this ideal is $$I_2:=\langle y_3y_6, y_2y_6, y_2y_3y_7, y_1y_2y_4, y_1y_4y_6y_7, y_1y_3y_4y_7\rangle,$$ with $S/I_2$ having the minimal graded free resolution
      $$0\rightarrow S^3(-6)\rightarrow S(-3)\oplus S^2(-4)\oplus S^5(-5)\rightarrow S^2(-2)\oplus S^2(-3)\oplus S^2(-4)\rightarrow S\rightarrow S/I_2\rightarrow 0.$$
\end{itemize}

As we can see, $t_2(I_1)=4$ and $t_2(I_2)=3$, obviously different numbers.

As a side note, in terms of Hilbert series (and functions) of the corresponding algebras, by \cite{Mac}, we have equality throughout:

$$HS({\rm OT}(\C^{\perp}),s)=HS(S/I_1,s)=HS(S/I_2,s).$$ Furthermore, it is worth mentioning that, in general, the Hilbert series of the Orlik-Terao algebra equals the Poincar\'{e} polynomial of the hyperplane arrangement evaluated at $s/(1-s)$; see \cite{Te} (for $\mathbb K$ a field of characteristic 0) and \cite{Be} (for any field $\mathbb K$). In our example, $\displaystyle HS({\rm OT}(\C^{\perp}),s)=\frac{1+3s+4s^2+3s^3}{(1-s)^4}$.
\end{exm}

\begin{thm}\label{thm13} Let $\mathcal C$ be an $[n,k,d]$-linear code, with $d\geq 2$. Then, the second generalized Hamming weight satisfies

$$t_2({\rm IOT}(\C^{\perp}))+1\leq d_2(\C)\leq t_2({\rm IOT}(\C^{\perp}))+2.$$
\end{thm}
\begin{proof} Let $H$ be a parity-check matrix of $\C$; this is an $(n-k)\times n$ matrix of rank $n-k$. From Section \ref{sec:genHam}, we have $$d_2(\mathcal C)=\min_{I\subseteq[n]}\{|I|\,|\, |I|-\dim_{\mathbb K}{\rm Span}_{\mathbb K}\{H_i\,|\, i\in I\}\geq 2\}.$$ For convenience, we denote $H_I:={\rm Span}_{\mathbb K}\{H_i\,|\, i\in I\}$, for any $I\subseteq [n]$. Also, we will drop the subscript $\mathbb K$.

Let $m=d_2(\C)$. If $m=2$, then, since $1\leq d=d_1(\C)<d_2(\C)=m$, forces $d=1$; contradiction. So $m\geq 3$.

After some reordering of the columns of $H$, suppose $J=\{1,\ldots,m\}$ has the property that $|J|-\dim(H_J)=m-\dim(H_J)\geq 2$.

Since $m$ is minimum, for any $i\in J$ we have $|J\setminus\{i\}|-\dim(H_{J\setminus\{i\}})\leq 1$, which gives $|J|-\dim(H_{J\setminus\{i\}})\leq 2$. From $2\leq |J|-\dim(H_J)\leq 2+\dim(H_{J\setminus\{i\}})-\dim(H_J)$, we obtain $\dim(H_J)=\dim(H_{J\setminus\{i\}})$, for all $i\in J$. So $\dim(H_{J\setminus\{i\}})\leq m-2$, and since $|J\setminus\{i\}|=m-1$, we obtain that $J\setminus\{i\}$ is a dependent set for every $i\in J$.

Consider $i=1$, and let $C_1$ be a maximal size circuit in $J\setminus\{1\}$. Suppose $C_1=\{2,\ldots,s\}$, for some $s\leq m$. Let $C_2$ be a maximal size circuit in $J\setminus\{2\}$. Since $C_2\nsubseteq C_1$ and $C_1\nsubseteq C_2$, we have two cases:

\begin{itemize}
  \item[(i)] $C_1\cap C_2\neq \emptyset$. In this case, we suppose $C_2=\{1,3,\ldots,3+u,s+1,\ldots,s+v\}$, for some $0\leq u\leq s-3$ and $0\leq v\leq m-s$.
  \item[(ii)] $C_1\cap C_2=\emptyset$.
\end{itemize}

\underline{Case (i)}. Suppose the dependency of $C_1$ is $$D_1:= a_2y_2+y_3+a_4y_4+\cdots+a_sy_s,$$ with $a_i\in\mathbb K\setminus\{0\}$, and the dependency of $C_2$ is $$D_2:= b_1y_1+y_3+b_4y_4+\cdots+b_{3+u}y_{3+u}+b_{s+1}y_{s+1}+\cdots+b_{s+v}y_{s+v},$$ with $b_j\in\mathbb K\setminus\{0\}$.

Then, after some possible reordering of the columns of $H$, we can assume that $D_3:=D_1-D_2$ becomes the linear dependency that corresponds to a circuit in $J\setminus\{3\}$:

$$-b_1y_1+a_2y_2+(a_t-b_t)y_{t}+\cdots+(a_{3+u}-b_{3+u})y_{3+u}+a_{4+u}y_{4+u}+\cdots+a_sy_s-b_{s+1}y_{s+1}-\cdots -b_{s+v}y_{s+v},$$ for some $4\leq t\leq 3+u$, with $a_t\neq b_t$.

We will use the following convenient notation: if $D:=c_{i_1}y_{i_1}+\cdots+c_{i_k}y_{i_k}$ is a dependency corresponding to a circuit, then we can write $$\frac{\partial(D)}{y_{i_1}\cdots y_{i_k}}=\frac{c_{i_1}}{y_{i_1}}+\cdots+\frac{c_{i_k}}{y_{i_k}}.$$

With this, if in the expression $D_3=D_1-D_2$ we formally replace $y_i$ with its reciprocal $1/y_i$, we obtain

$$\frac{\partial(D_3)}{y_1y_2y_t\cdots y_{s+v}}=\frac{\partial(D_1)}{y_2\cdots y_s}-\frac{\partial(D_2)}{y_1y_3\cdots y_{3+u}y_{s+1}\cdots y_{s+v}}.$$ Clearing denominators we obtain the 1-syzygy

$$y_3\cdots y_{t-1}\partial(D_3)-y_1y_{s+1}\cdots y_{s+v}\partial(D_1)+y_2y_{4+u}\cdots y_s\partial(D_2)=0, $$ of degree $s+v-1$. If it happens that $\partial(D_3)=P_1\partial(D_1)+P_2\partial(D_2)$, for some $P_1,P_2\in S$ of the appropriate degrees (i.e., the syzygy above is trivial), then we immediately obtain the non-trivial 1-syzygy on the minimal generators $\partial(D_1)$ and $\partial(D_2)$:
$$(y_3\cdots y_{t-1}P_1-y_1y_{s+1}\cdots y_{s+v})\partial(D_1)+(y_3\cdots y_{t-1}P_2+y_2y_{4+u}\cdots y_s)\partial(D_2)=0,$$ which, after some possible simplifications has degree $\leq s+v-1$.

From both cases we conclude $$t_2({\rm IOT}(\C^{\perp}))\leq s+v-1\leq s+(m-s)-1=m-1=d_2(\C)-1.$$

\underline{Case (ii)}. Suppose $|C_2|=s'-1\leq m-1$. Since $|C_1|=s-1\leq m-1$, and $C_1\cap C_2=\emptyset$, then $(s-1)+(s'-1)\leq m$.

If $D_1$ is the dependency of $C_1$ and $D_2$ is the dependency of $C_2$, then $\partial(D_1)\in{\rm IOT}(\C^{\perp})$ is of degree $s-2$ and $\partial(D_2)\in{\rm IOT}(\C^{\perp})$ is of degree $s'-2$. The Koszul syzygy $$\partial(D_2)\cdot\partial(D_1)-\partial(D_1)\cdot\partial(D_2)=0$$ produces a 1-syzygy of ${\rm IOT}(\C^{\perp})$ of degree $(s-2)+(s'-2)$. This implies $$t_2({\rm IOT}(\C^{\perp}))\leq s+s'-4\leq m-2=d_2(\C)-2.$$

Both cases above give the inequality $$t_2({\rm IOT}(\C^{\perp}))+1\leq d_2(\C).$$

\bigskip

For the other inequality, first denote $t:=t_2({\rm IOT}(\C^{\perp}))$. Suppose there is a 1-syzygy of degree $t$ on some minimal generators $\partial(E_1),\ldots,\partial(E_s)$ of the Orlik-Terao ideal of $\C^{\perp}$:

$$(*) \, P_1\partial(E_1)+\cdots+P_s\partial(E_s)=0,$$ with $P_1,\ldots,P_s$ nonzero homogeneous polynomials in $S$, and such that $\deg(P_i)+\deg(\partial(E_i))=t$, for all $i=1,\ldots,s$.

Each $E_i$ is the minimal dependency corresponding to a circuit $J_i$ of size $|J_i|=\deg(\partial(E_i))+1$, which we will denote with $e_i$.

\medskip

\noindent\underline{Claim:} For $i\neq j$, since $J_i$ and $J_j$ are distinct circuits, then $|J_1\cup J_2|\geq d_2(\C)$.

\underline{Proof of Claim:} We have $|J_i|-\dim(H_{J_i})\geq 1$ and $|J_j|-\dim(H_{J_j})\geq 1$. So,

$$|J_i\cup J_j|+|J_i\cap J_j|=|J_i|+|J_j|+2+\dim(H_{J_i})+\dim(H_{J_j}).$$ Since $\dim(H_{J_i})+\dim(H_{J_j})\geq \dim(H_{J_i\cup J_j})+\dim(H_{J_i\cap J_j})$, we then have
$$|J_i\cup J_j|-\dim(H_{J_i\cup J_j})\geq 2-(|J_i\cap J_j|-\dim(H_{J_i\cap J_j})).$$ But $J_i\cap J_j\subsetneq J_i$ and $J_i\cap J_j\subsetneq J_j$, and so, $J_i\cap J_j$ is either the empty set or it is an independent set. In either case, we have $|J_i\cap J_j|=\dim(H_{J_i\cap J_j})$, and therefore, $$|J_i\cup J_j|-\dim(H_{J_i\cup J_j})\geq 2.$$ From Wei's formula (see Section \ref{sec:genHam}), we have $|J_i\cup J_j|\geq d_2(\C)$.

\medskip

If $M\in S$ is a monomial, then define {\em the support of $M$} to be $${\rm supp}(M):=\{i\, |\, y_i \mbox{ divides } M\}.$$

Let us go back to the syzygy $(*)$. Suppose $J_1=\{1,\ldots,e_1\}$. Suppose we picked a monomial order on $S$ with $y_1>\cdots>y_n$. Then, the leading monomial of $P_1\partial(E_1)$, which is $M_1y_1\cdots y_{e_1-1}$, in order to get cancelled, it must occur as some monomial in the expressions of at least one of $P_2\partial(E_2),\ldots,P_s\partial(E_s)$. Say it shows up in $P_2\partial(E_2)$. Suppose $J_2=\{i_1,\ldots,i_{e_2}\}$, and suppose that this monomial is $M_2y_{i_1}\cdots \widehat{y_{i_{\nu}}}\cdots y_{i_{e_2}}$, for some $\nu\in\{1,\ldots,{e_2}\}$.

We have $$(J_1\cup J_2)\setminus\{e_1, i_{\nu}\}\subseteq (J_1\setminus\{e_1\})\cup(J_2\setminus\{i_{\nu}\})\subseteq {\rm supp}(M_1y_1\cdots y_{e_1-1})={\rm supp}(M_2y_{i_1}\cdots \widehat{y_{i_{\nu}}}\cdots y_{i_{e_2}}).$$

This gives $$|J_1\cup J_2|-2\leq |(J_1\cup J_2)\setminus\{e_1, i_{\nu}\}|\leq |{\rm supp}(M_1y_1\cdots y_{e_1-1})|\leq \deg(M_1y_1\cdots y_{e_1-1})=t.$$ Therefore, from Claim above we have the inequality $$d_2(\C)\leq t+2.$$
\end{proof}

\begin{exm}\label{exm15} If in Example \ref{exm10} we showed that the lower bound is attained, in this example we show that the upper bound is attained.

Let $\mathcal C$ be the linear code over a field $\mathbb K$, of characteristic not 2, with a parity-check matrix

$$H=\left[\begin{array}{rrrrrr}1&0&0&1&1&0\\ 0&1&0&-1&0&1\\ 0&0&1&0&-1&-1\end{array}\right].$$ Some of the circuits are $$\{1,2,4\}, \{1,3,5\}, \{4,5,6\}, \{2,3,6\},$$ with corresponding dependencies:

$$D_1:=-y_1+y_2+y_4, D_2:=-y_1+y_3+y_5, D_3:=y_4-y_5+y_6, D_4:=-y_2+y_3+y_6.$$

The corresponding elements in the Orlik-Terao ideal are

\begin{eqnarray}
\partial(D_1)&=&-y_2y_4+y_1y_4+y_1y_2\nonumber\\
\partial(D_2)&=&-y_3y_5+y_1y_5+y_1y_3\nonumber\\
\partial(D_3)&=&y_5y_6-y_4y_6+y_4y_5\nonumber\\
\partial(D_4)&=&-y_3y_6+y_2y_6+y_2y_3.\nonumber
\end{eqnarray} In fact, one can check that they minimally generate this ideal.

It is not hard to check that on these elements we have a linear syzygy:

$$(-y_3-y_5)\partial(D_1)+(y_2+y_4)\partial(D_2)+(-y_2+y_3)\partial(D_3)+(-y_4+y_5)\partial(D_4)=0.$$ So we have $t_2({\rm IOT}(\mathcal C^{\perp}))=3$.

Now we compute $d_2(\C)$. A generator matrix for $\C=(\C^{\perp})^{\perp}$ is

$$G=\left[\begin{array}{rrrrrr}-1&1&0&1&0&0\\ -1&0&1&0&1&0\\ 0&-1&1&0&0&1\end{array}\right].$$ From Remark \ref{rem1.1} (iii), we have $d_2(\C)$ equals 6 minus the maximum number of columns of $G$ that span a $3-2=1$ vector space: $d_2(\C)=6-1=5$.

Hence, for this example, $$d_2(\C)=t_2({\rm IOT}(\mathcal C^{\perp}))+2.$$
\end{exm}

\begin{rem}\label{lastrem} Regarding the higher generalized Hamming weights we know the following.

By \cite{PrSp}, the Orlik-Terao algebra ${\rm OT}(\C^{\perp})$ is Cohen-Macaulay of projective dimension ${\rm pdim}_S({\rm OT}(\C^{\perp}))={\rm ht}({\rm IOT}(\C^{\perp}))=k$.

Also because of the Cohen-Macaulay property, we have ${\rm reg}({\rm OT}(\C^{\perp}))=T_k({\rm IOT}(\C^{\perp}))-k$.

At the same time, by \cite[Propositions 1.2 and 2.7]{GaSiTo}, we have ${\rm reg}({\rm OT}(\C^{\perp}))=n-k-c$, where $c\geq 1$ is the number of components of $\mathcal A(\C^{\perp})$. Therefore, $$n = T_k({\rm IOT}(\C^{\perp}))+c\geq t_k({\rm IOT}(\C^{\perp}))+c.$$

The dual version of \cite[Corollary 4.3]{JoVe1} says that, if $\C$ is an $[n,k,d]$-linear code, then $$d_k(\C)=n-\ell,$$ where $\ell$ is the number of zero columns in a parity-check matrix of $\C$ (i.e., the number of loops in the matroid of $\C^{\perp}$). Since we assumed that $d\geq 2$, we have $\ell=0$. Everything together leads to $$d_k(\C)\geq t_k({\rm IOT}(\C^{\perp}))+c.$$

A relevant example is the following. Let $\mathcal C$ be the linear code over a field $\mathbb K$, of characteristic not 2, with a parity-check matrix

$$H=\left[\begin{array}{rrrrrrr}1&0&0&0&-1&0&0\\ 0&1&0&0&-1&0&0\\ 0&0&1&0&0&-1&1\\ 0&0&0&1&0&-1&0\end{array}\right].$$ The defining linear forms for $\mathcal A(\C^{\perp})$ are $x_1, x_2, x_3, x_4, -(x_1+x_2), -(x_3+x_4), x_4\in \mathbb K[x_1,\ldots,x_4]$; it is clear that there are $c=2$ components.

The circuits of $H$ are $\{1,2,5\}, \{3,4,6\}, \{3,7\}, \{4,6,7\}$ that will give the minimal generators of ${\rm IOT}(\C^{\perp})\subset S:=\mathbb K[y_1,\ldots,y_7]$:
$$y_2y_5+y_1y_5+y_1y_2, y_4y_6+y_3y_6+y_3y_4, y_7-y_3.$$ With these generators we obtain the minimal graded free resolution:

$$0\rightarrow S(-5)\rightarrow S^2(-3)\oplus S(-4)\rightarrow S(-1)\oplus S^2(-2)\rightarrow S\rightarrow {\rm OT}(\C^{\perp})\rightarrow 0.$$ We have $$t_1({\rm IOT}(\C^{\perp}))=1, t_2({\rm IOT}(\C^{\perp}))=3, t_3({\rm IOT}(\C^{\perp}))=5=T_3({\rm IOT}(\C^{\perp})).$$

The Stanley-Reisner ideal of $H$ is $I=\langle y_1y_2y_5, y_3y_4y_6, y_3y_7, y_4y_6y_7\rangle$, and if we use this to compute the generalized Hamming weights of $\C$ we obtain:

$$d_1(\C)=2=t_1({\rm IOT}(\C^{\perp}))+1, d_2(\C)=4=t_2({\rm IOT}(\C^{\perp}))+1, d_3(\C)=7=t_3({\rm IOT}(\C^{\perp}))+2.$$

\medskip

At this moment, we can only conjecture that if $\mathcal C$ is an $[n,k,d]$-linear code with $d\geq 2$, then, for any $1\leq r\leq k$, we have $$d_r(\C)\geq t_r({\rm IOT}(\C^{\perp}))+1.$$
\end{rem}

\vskip 0.2in

\renewcommand{\baselinestretch}{1.0}
\small\normalsize 

\bibliographystyle{amsalpha}

\end{document}